\documentclass[aps,pra,showpacs,twoside,twocolumn,10pt,nofootinbib]{revtex4-2}
\usepackage[colorlinks=true, citecolor=red, urlcolor=blue]{hyperref}
\usepackage{epsfig,newlfont,amssymb,amsfonts,amsmath,bm,subfigure,palatino,mathtools,amsthm,braket,soul,enumitem,color,graphics,graphicx,times,physics,bbold,comment, braket, ulem, qcircuit}
\usepackage[thinc]{esdiff}
\usepackage{xcolor}
\usepackage{dsfont}
\usepackage{float}
\usepackage{tabularx}
\usepackage{array}
\usepackage[export]{adjustbox}
\usepackage{mathtools}
\newtheorem{theorem}{Theorem}

\newtheorem{lemma}{Lemma}

\newcommand{\ph}[1]{{\color{magenta}#1}}

\newcolumntype{Y}{>{\centering\arraybackslash}X}
\begin{document}
\title{
Efficient formulation of quantum network  under amplitude damping noise: Highlighting benefits over its Pauli-twirled counterpart
} 

\author{Sudipta Mondal$^1$, Pritam Halder$^{1}$, Stav Haldar$^2$, Aditi Sen(De)$^1$}

\affiliation{$^1$Harish-Chandra Research Institute, A CI of Homi Bhabha National Institute, Chhatnag Road, Jhunsi, Prayagraj 211019, India\\
\(^2\) Manning College of Information and Computer Sciences,
University of Massachusetts Amherst,
140 Governors Drive, Amherst, MA, USA }

\begin{abstract}

At the heart of building a large-scale quantum internet lies the challenge of establishing long-distance entanglement using quantum repeaters, which mitigate direct transmission losses but introduce additional noise in the nodes via interactions with the environment and imperfect operations. This effect has typically been studied under a simplifying Pauli channel assumption. Our study focuses on distributing end-to-end entanglement in a homogeneous, repeater-based linear quantum network operating under a non-Pauli noise, specifically, amplitude damping noise, which we refer to as amplitude damping-affected quantum network (AQN). Unlike its twirled counterpart (TAQN), where the resulting state is fully Bell-diagonal with a single parameter, we prove that the AQN produces a block-diagonal state in the Bell basis with four parameters.  We develop a method for the simulation of AQN, where we keep track of these four parameters of each entangled link, along with the number of times noise acts on it, i.e., its age, until it is consumed for swapping. Our results reveal that across diverse policies, including \textsc{nesting} and \textsc{swap-asap}, AQN consistently outperforms TAQN in terms of both fidelity and average entanglement. The benefit is most significant in the low-probability regime of elementary link generation, highly relevant for near-term experiments. Notably, we also identify the coherence-time and link-probability regions where TAQN fails while AQN succeeds in distributing end-to-end entanglement.

\begin{comment}
    {Quantum technologies such as communication, sensing, and distributed computation will rely on robust networks capable of distributing entanglement between distant nodes. In this work, we compare homogeneous linear repeater chains subject to two distinct noise models: the amplitude damping quantum network (AQN) and the twirled amplitude damping quantum network (TAQN). A key distinction lies in the parametrization of the states: while TAQN remains restricted to a single-parameter Bell-diagonal family, AQN states require four parameters but retain this structure under arbitrary sequences of noise and swap operations. Our results show that, across different chain lengths, AQN consistently outperforms TAQN in terms of both average entanglement and the end-to-end fidelity of the final shared state. The advantage is particularly pronounced in the low-probability regime of elementary link generation, a setting of immediate relevance for near-term experimental implementations.}
\end{comment}

\end{abstract}

\maketitle

% \ph{Points --
% \begin{enumerate}
%     \item Where to write the Standard deviation?
% \end{enumerate}
% }

\section{Introduction}

The quantum internet \cite{Kimble2008, VanMeter2014, Wehner2018, Cacciapuoti2020, Dowling2020, Illiano2022, Munro2022, quant_internet_rev} promises efficient transmission of quantum information alongside classical messages among multiple parties, enabling distributed applications \cite{jozsa2000, troupe2022, SH_QCS_PRA,SH_QCS2_PRA, Ducoing2025, cirac1999, ge2018, proctor2018, Pirker2019, Markham2008,Gottesman2012, Kellerer2014} far beyond the scope of today’s classical internet. Its realization relies on the efficient distribution of long-distance entanglement at high rates, for example, by sending single-photon qubits through optical fibers, free-space, or atmospheric channels. However, such direct transmission suffers from an exponentially decreasing success probability with distance  \cite{ Karp1988, Svelto2010, Kaushal2017}. To overcome this, {\it quantum repeaters} \cite{Briegel1998, Dur_1999, Muralidharan2016} have been proposed, which divide the total length into shorter segments, generate elementary links by performing distillation or error correction \cite{Bennett1996_dis, Deutsch1996} together with entanglement swapping \cite{bennet1993_tel, mode1993} at intermediate nodes to connect these segments \cite{Briegel1998, Dur_1999, Muralidharan2016}. Recent experimental progress toward building a quantum internet includes the realization of three-node quantum networks \cite{Pompili2021, Hermans_2022}, proof-of-principle demonstrations of repeater protocols \cite{Bhaskar2020, langenfeld2021}, and heralded entanglement generation across metropolitan-scale fiber links \cite{Knaut2024, Stolk2024}. Yet, decoherence in quantum memories, probabilistic link generation, and imperfect entanglement swapping continue to limit the achievable rate, quality, and distance of entangled links. To mitigate these effects, concrete theoretical frameworks and protocols have been proposed \cite{kamin2023, Khatri2022, halder_stav, Haldar2025, Iesta2023, lsuDESIGNLONGDISTANCE}, typically focusing on entanglement distribution under quantum memory decoherence with Pauli-type noise models \cite{Nielsen2000QuantumComputation}.

A key advantage of simulating a Pauli noise-affected quantum network (PQN) lies in its simplicity: at any stage of the simulation, every distributed entangled pair remains diagonal in the Bell basis, with coefficients determined solely by how long the pair has been stored in memory. In other words, the age of any `link' uniquely specifies its fidelity and the amount of entanglement. Moreover, when two noisy links undergo entanglement swapping, the resulting link acquires an age equal to the sum of ages of the consumed links \cite{schmidt2020, Khatri2022, Haldar2025}. Thus, simulations only require tracking a single parameter per link, the age, rather than the full density matrix. This efficiency explains why most existing studies \cite{kamin2023, Khatri2022, halder_stav, Haldar2025, Iesta2023, lsuDESIGNLONGDISTANCE} on quantum repeater chains focus on Pauli channels, while non-Pauli channels are usually approximated by Pauli twirling \cite{dur2005, Emerson2007, silva2008, Dankert2009} (cf. also \cite{Sarvepalli2009, Ghosh2012}). However, this simplification undermines the accuracy of optimization procedures, including swapping and cut-off policies, distillation protocols, and encoding/decoding strategies \cite{halder_stav,deandrade2024,kamin2023}, and fails to capture the genuine effects of non-Pauli channels such as amplitude damping (AD) noise in realistic systems.

Unlike most existing studies, this work goes beyond the usual Pauli approximations by directly addressing the non-twirled amplitude damping channel (ADC), a crucial step, as twirling is also resource-intensive and often impractical to implement \cite{Anwar2005, lopez2010}.
By retaining the full quantum description of the states, our approach establishes a foundation for more effective optimization of protocols and policies under realistic noise conditions. In the non-Pauli noise regime, we investigate a natural question: ``Can network simulations be performed without tracking the full density matrices of all two-party states?" We affirmatively answer this by exploring the amplitude damping noise, relevant to quantum memories realized in atomic systems, quantum dots, and photonic architectures \cite{Nielsen2000QuantumComputation,Chirolli2008,Khatri2022}. Specifically, for {\it an amplitude damping-affected linear repeater network, referred to as AQN} (see Fig.~\ref{fig:schematic}), we address a more precise question:
What is the most general state of the AQN  at any arbitrary step in generating an end-to-end entangled link? To answer this, we derive noise update and entanglement swap rules within the repeater protocol and prove that the resulting state indeed admits a simplified parametrization: it is block-diagonal in the Bell basis and fully specified by four parameters. Moreover, we exhibit that even after multiple successive rounds of noise and swapping, the state consistently retains this structure, remaining confined within the same four-parameter family. In contrast, when adopting the Pauli-twirled approximation of AD noise, termed as TAQN, the resulting state remains strictly diagonal in the Bell basis and is effectively described by a single parameter, thus carrying less information than the corresponding AQN state.
Indeed, through simulations of various distribution strategies such as \textsc{swap-at-last}, \textsc{nesting} \cite{Briegel1998, Dur_1999}, and \textsc{swap-asap} \cite{TJCoopmans2021,Shchukin2019, shchukin2022, kamin2023} and assuming deterministic swapping, we explicitly demonstrate that accounting for amplitude damping noise yields higher entanglement values and fidelities at the end of the protocol, compared to its twirled counterpart. Specifically, we pinpoint the regions of link-generation probability and coherence time where AQN sustains a nonvanishing entanglement generation capability, while its TAQN version completely fails to produce entanglement -- we call this a region of {\it absolute-advantage}.

The remainder of this work is organized as follows. Sec.~\ref{sec:main_proofs} presents the theoretical formulation of the amplitude damping-affected quantum network in a linear repeater network. In particular, here, we derive the noise update and swap rules for arbitrary links of the AQN (Sec.~\ref{sec:rules}) and discuss the same for the Pauli-twirled counterpart of the ADC in Sec.~\ref{sec:pauli-twirled}. Sec.~\ref{sec:sim_method} describes the simulation methodology employed to analyze the AQN under different policies, while Sec.~\ref{sim_res} reports the corresponding simulation results. Sec.~\ref{sec:con} contains the concluding remarks.

\begin{figure*}
    \centering
\includegraphics[width=1.0\linewidth]{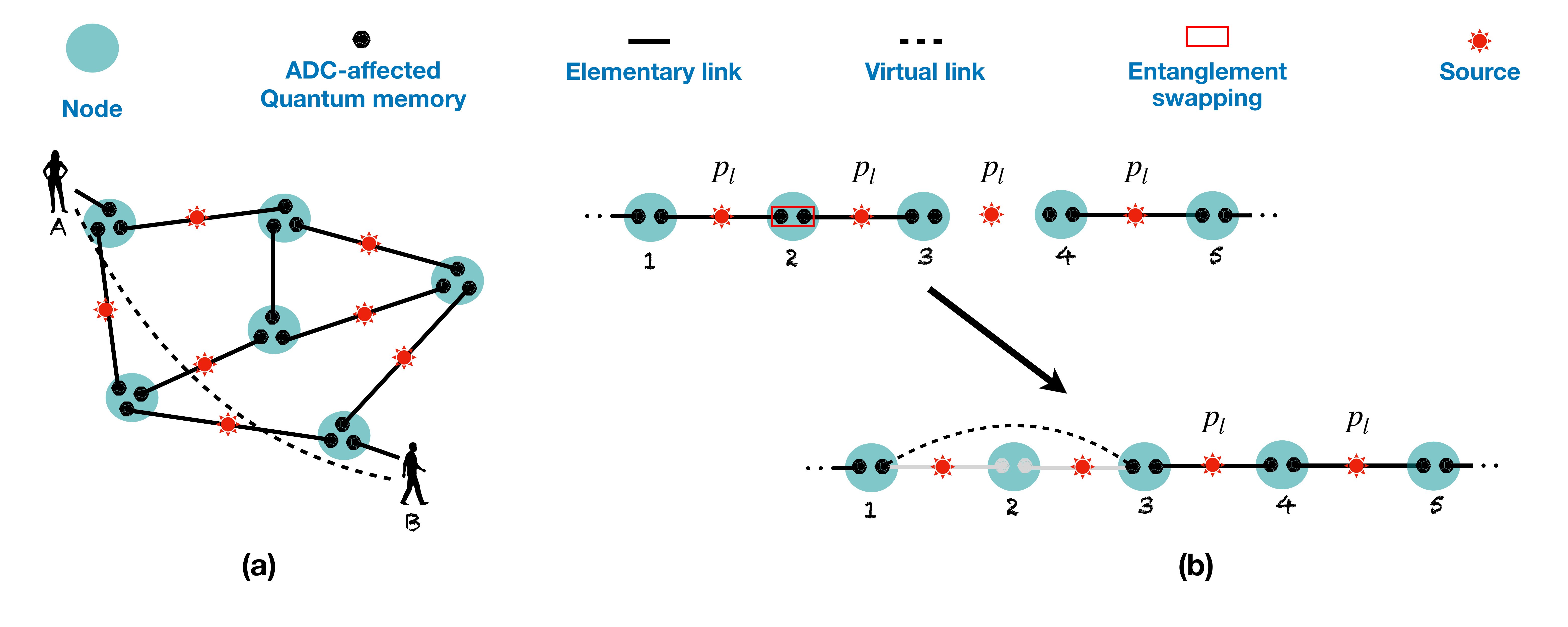}
    \caption{(a) A schematic representation of a general quantum network having an arbitrary number of nodes (big circles) where the quantum memories are affected by amplitude damping noise (solid circles). Solid lines represent elementary-level links generated from the sources (stars), while the dashed lines are the virtual links produced through entanglement swapping, marked by a rectangular box. The ultimate goal is to establish an entangled link between A and B, marked as a dashed line.
    (b) An example of a part of a long amplitude damping-affected linear quantum network (AQN). Here, the repeater chain is homogeneous, i.e., all the sources have the same link generation probability $p_l$. At a given instance, all links are assumed to be 
    successfully generated except the one between the nodes $3$ and $4$. Since both memories at node $2$ are active, a swapping operation (which we assume to be deterministic throughout our work) at this node establishes a virtual entangled link between the nodes $1$ 
    and $3$. Simultaneously, the direct link between nodes $3$ and $4$ is successfully 
    generated, extending the connectivity of the network.
}
\label{fig:schematic}
\end{figure*}

\section{Modelling of amplitude damping noise-affected linear quantum network}
\label{sec:main_proofs}

This section begins by presenting the setup of a linear quantum network and analyzing the impact of amplitude damping noise on an elementary link. Building on this foundation, we present the core results of our work: the characterization of a generic link in AQN and the derivation of the AD noise update and swap rules. These results serve as the backbone for the numerical simulations in Sec.~\ref{sec:policies}, which offer detailed insights into optimizing end-to-end entanglement creation in quantum repeater networks subject to non-Pauli noise.

\subsection{A linear quantum network}
Let us consider a quantum network with $N$ nodes arranged in a line, as shown in Fig.~\ref{fig:schematic}(b). Each end node of a linear chain consists of a single quantum memory, while all intermediate nodes, referred to as repeaters, possess two quantum memories. We deal with a generic process by which entanglement is generated with some probability $p_l\in[0,1]$ per attempt between nearest neighbor nodes. Note that $p_l$ depends on the hardware parameters such as detector inefficiencies, transmission probability of channels, and 
finite absorption cross sections of quantum memories \cite{Khatri2021, Khatri2022}. We refer to these entangled pairs between nearest neighbors as active elementary links. An elementary link is inactive if no entanglement exists between the corresponding nodes. Each repeater node is allowed at most two connections -- one to a node ahead and one to a node behind along the chain, thereby avoiding leapfrogged states. An entangled link distributed between non-nearest neighbor nodes (via entanglement swapping operations) is referred to as a virtual link. 
\subsection{Elementary link and effect of ADC}
We assume that when a fresh elementary link is generated, it is of the form of $\ket{\Phi^+}=(\ket{00}+\ket{11})/\sqrt 2$. Let $\sigma_{PQ}=\Phi^+=\ketbra{\Phi^+}$ be the density matrix representing the elementary link generated between two quantum memories, $P$ and $Q$, which can decohere with time. Hence, in a discrete time setting \cite{Ghosh2012} which is the model considered in this work, the effect of ADC acting for $m_1$ time steps on $\sigma_{PQ}$ can be described as (see Appendix~\ref{app:noise_on_pauli_basis})
\begin{eqnarray}
    \nonumber \sigma_{PQ}(m_1)&=&(\mathcal N_P^{\circ {m_1}}\otimes \mathcal N_{Q}^{{\circ m_1}})\sigma_{PQ},\\\nonumber&=&\frac{1}{4}[\mathbb 1\otimes \mathbb 1+a( X\otimes  X- Y\otimes  Y)\\\nonumber &&+(1- a)(\mathbb 1\otimes  Z+ Z\otimes \mathbb 1)\\&&+\{(1- a)^2+ a^2\} Z\otimes  Z],
   \label{eq:ele_adc}
\end{eqnarray}
where $a=(1-\gamma)^{m_1}$ with \(\gamma\) being the strength of the noise parameter and $\{\mathbb 1,  X,  Y,  Z\}$ is the set of a single-qubit Pauli basis. Defining $\Phi_{xz}=\ketbra{\Phi_{xz}}$,   $\Phi_{x'z'}^{xz}=\ketbra{\Phi_{xz}}{\Phi_{x'z'}}$ and rewriting Eq.~\eqref{eq:ele_adc} in terms of Bell basis, $\{\Phi_{xz}=\mathtt{Proj}\left(\mathbb 1\otimes X^xZ^z\ket{\Phi^+}\right)\}_{x,z=0}^1$~\footnote{The Bell basis consists of  \(\ket{\Phi_{00}}=\ket{\Phi^+}\), \(\ket{\Phi_{01}}=\ket{\Phi^-}=(\ket{00}-\ket{11})/\sqrt{2}\), \(\ket{\Phi_{10}}=\ket{\Psi^+}=(\ket{01}+\ket{10})/\sqrt{2}\), and \(\ket{\Phi_{11}}=\ket{\Psi^-}=(\ket{01}-\ket{10})/\sqrt{2}\).}, we arrive at the following Lemma:
\begin{lemma}
    A link $\sigma_{PQ}=\Phi^+$, affected by amplitude damping noise $m_1$ times, can be expressed in the Bell basis as 
    \begin{eqnarray}
        \nonumber \sigma_{PQ}(\{c_i\}_{i=1}^3)&=&c_1\Phi_{00}+c_2\Phi_{01}+c_3(\Phi_{01}^{00}+\Phi_{00}^{01})\\&&+2^{-1}(1-\sum_{i=1}^2c_i)(\Phi_{10}+\Phi_{11}),
    \end{eqnarray}
    \rm{where} \(c_1={(1+a^2)}/{2}\), \(c_2={(1-a)^2}/{2}\) and \(c_3={(1-a)}/{2}\), \rm{with} $a=(1-\gamma)^{m_1}$.
    \label{lem:adc_link_bell}
\end{lemma}
The details of the proof of Lemma~\ref{lem:adc_link_bell} are given in Appendix~\ref{app:noise_on_pauli_basis}.
Notice that, unlike the case for Pauli channels, where the noise-affected state remains Bell diagonal, $\sigma_{PQ}(\{c_i\}_{i=1}^3)$ is block-diagonal in the Bell basis with three parameters $\{c_i\}_{i=1}^3$. Further, any link in the AQN, characterized by four coefficients has the following structures: the first, second, third, and the fourth coefficients of any link correspond to the terms $\Phi_{00},\Phi_{01},(\Phi_{01}^{00}+\Phi_{00}^{01})$, and $(\Phi_{10}^{11}+\Phi_{11}^{10})$, respectively. The first block consists of \(\Phi_{00}\) and \(\Phi_{01}\) states with the diagonal coefficients \(c_1\) and \(c_2\) and the equal off-diagonal ones, \(c_3\). The second block is diagonal, having equal coefficient $2^{-1}(1-c_1-c_2)$.

\subsection{Two elementary level AD noisy links and their swap rule}
Let us consider two noisy elementary links, \(\sigma_{PQ_1}\) and \(\sigma_{Q_2R}\), each defined by its three parameters, as discussed in Lemma~\ref{lem:adc_link_bell}. Now the usual entanglement swapping operation performed on \(Q_1\) and \(Q_2\) of \(\rho_{PQ_1Q_2R}=\sigma_{PQ_1}\otimes \sigma_{Q_2R}\), producing \(\rho_{PR}\) can be represented as \cite{Khatri2022}
\begin{eqnarray}
    \nonumber &&\mathcal S_{PQ_1Q_2R\to PR}(\rho_{PQ_1Q_2R})\\\nonumber&&=\sum_{x,z=0}^1(\bra{\Phi_{xz}}_{Q_1 Q_2}\otimes \mathsf C_R)(\rho_{PQ_1Q_2R})(\ket{\Phi_{xz}}_{Q_1 Q_2}\otimes  \mathsf C_R^T),\\
\end{eqnarray}
where the Bell basis measurement (BSM) is performed on the qubits \(\{Q_1,Q_2\}\) and $\mathsf C_R= X^x_RZ^z_R$ is the local unitary operation, applied to the end qubit $R$ based on the measurement outcomes, and $(\cdot)^T$ is the transpose operation. 
\begin{lemma}
After performing the swapping operation $ \mathcal S_{PQ_1Q_2R\to PR}$  on two elementary level noisy links, \(\sigma_{PQ_1}(\{c_i\}_{i=1}^3)\) and \(\sigma_{Q_2R}(\{d_i\}_{i=1}^3)\), the resultant state \(\rho_{PR}(\{f_i\}_{i=1}^4)\) belongs to a four-parameter family, given by
\begin{eqnarray}
    &&\nonumber \rho_{PR}(\{f_i\}_{i=1}^4)\\\nonumber &&=f_1\Phi_{00}+f_2\Phi_{01}+f_3(\Phi_{01}^{00}+\Phi_{00}^{01})\\\nonumber&&~~~~+2^{-1}(1-\sum_{i=1}^2f_i)(\Phi_{10}+\Phi_{11})+f_4(\Phi_{11}^{10}+\Phi_{10}^{11}),\\
    \label{eq:1st_swap_state}
\end{eqnarray}
where the coefficients  are
\begin{eqnarray}
  \nonumber &&f_{k}=\frac{1}{4}\bigg(1+ \big(2\sum_{i=1}^2c_i-1\big)\big(2\sum_{i=1}^2d_i-1\big) +2(-1)^{k-1}  \mathbf{c}\mathbf{d}\bigg)\\\nonumber&&\text{when } k=1,2, \text{   and}\nonumber\nonumber\\&&f_3=c_3\sum_{i=1}^2d_i;~~~~f_4=c_3(1-\sum_{i=1}^2d_i),
  \label{eq:1st_swap_update_coeff}
\end{eqnarray}
\text{with} \(\mathbf{c}=c_1-c_2,\quad \mathbf{d}=d_1-d_2\).
% \ph{Sudipta: Write the coefficients $f_i$s in terms of $c_i$s and $d_i$s here.} 
    \label{th:swap_ele}
\end{lemma}

The detailed proof of the Lemma~\ref{th:swap_ele} is given in Appendix~\ref{app:swap_1st_round}. 

Simplifying further, we can write $\{f_i\}_{i=1}^4$ as \( f_1=\frac{1}{4}(1+2ab+\mathbb{ab}), \;
f_2=\frac{1}{4}(1-2ab+\mathbb{ab}), \; f_3=\frac{1}{4}(1-a)(1+\mathbb{b}), \;
f_4=\frac{1}{4}(1-a)(1-\mathbb{b})\) with $a=(1-\gamma)^{m_1}$, $b=(1-\gamma)^{m_2}$, \(\mathbb{a}=a^2+(1-a)^2\) and  \(\mathbb{b}=b^2+(1-b)^2\). The above analysis shows that entanglement swapping of noisy elementary links produces states which are outside the three-parameter family of states mentioned in Lemma~\ref{lem:adc_link_bell}. Thus, the swap update rule (Lemma~\ref{th:swap_ele}) for two ADC-affected elementary links is distinct from the swap update rule of the Pauli noise scenario, where the age (single parameter) of the swapped link just becomes the sum of ages of the consumed links, i.e., Bell diagonal states remain Bell diagonal under entanglement swapping. This is not the case for ADC noise-affected states. Building on these observations leads us to the question of the general structure of these states, produced at an arbitrary time step. 

\subsection{General noise update and swap rule in AQN}
\label{sec:rules}
Let us first derive the update rule of the four-parameter virtual link
$\sigma_{PQ}(\{f_i\}_{i=1}^4)$ generated during the swap of some elementary noisy links, distributed by AD noise.
\begin{theorem}
(\textbf{AD noise update rule}) The action of ADC $n_1$ times on the link $\sigma_{PQ}(\{f_i\}_{i=1}^4)$ transforms the state to $\sigma_{PQ}(\{\tilde f_i\}_{i=1}^4)$ where the $f_i$s are updated to $\tilde f_i$s according to
\begin{eqnarray}
    \nonumber \tilde f_k&=&\frac{1}{4}[f_1(1+(-1)^{k-1}2 t+\mathbb t_+)+f_2(1+(-1)^k2t+\mathbb t_+)\\\nonumber &&+2f_3(1+\mathbb t_{-})+4f_4t(1-t)]\text{ for }k=1,2;\\
    \tilde f_3&=&\frac{1}{2}(1-t)+tf_3; ~~~~~~~\tilde f_4=f_4t,
    \label{eq:adc_general_state_update_rule}
\end{eqnarray}
with \(t=(1-\gamma)^{n_1}\),  \(\mathbb{t_{\pm}}=(1-t)^2\pm t^2\).
\label{th:arb_noise}
\end{theorem}
\begin{comment}
    \nonumber 
    \tilde f_3^1&=&\frac{1}{4}[(f_1+f_2)(1-\mathbb t_+)+2f_3\big(1-\mathbb{t_-}\big)\\\nonumber &&-4f_4t(1-t)];\\
\end{comment}
Therefore, starting from the four-parameter states defined in Lemma \ref{th:swap_ele}, the presence of AD noise for an arbitrary number of time steps updates the coefficients by the rule given in Eq.~\eqref{eq:adc_general_state_update_rule}, keeping the form of the quantum state preserved. 

Next, we calculate the SWAP rule for two four-parameter links.
\begin{comment}
    \begin{theorem}
    Consider two five-parameter links of the form $\sigma_{PQ_1}(\{f_i\}_{i=1}^5)$ and $\sigma_{Q_2R}(\{g_i\}_{i=1}^5)$. Performing SWAP $\mathcal S_{PQ_1Q_2R\to PR}$ creates a link $\rho_{PR}(\{h_i\}_{i=1}^5)$ which stays in the five-parameter family. The $h_i$s are calculated as
    \begin{eqnarray}
        \nonumber h_k &=& 4(s_1+s_2+(-1)^{k-1}2s_3)\text{ for }k=1,2;\\
        \nonumber h_3&=&4(s_1-s_2);\\ \nonumber
        h_l &=& 4(s_4 + (-1)^l s_5)\text{ for }l=4,5,
    \end{eqnarray}
    with \( s_1 = \frac{1}{16},
   s_{2} = \frac{1}{16}[(1 - 4f_{3}^1)(1 - 4g_{3}^1)],
   s_3 = \frac{1}{16}[(f_{1}^1 - f_{2}^1)(g_{1}^1 - g_{2}^1)],
   s_{4} = \frac{1}{8}(f_{4}^{1} + f_{5}^{1}),\) and \( 
 s_5 = \frac{1}{8}[(f_{4}^1 - f_{5}^1)(1 - 4g_{3}^1)]\).
 \label{th:gen_swap_rule}
\end{theorem}
\end{comment}
\begin{theorem}
   (\textbf{Swap rule with AD noise}) Consider two four-parameter links $\sigma_{PQ_1}(\{f_i\}_{i=1}^4)$ and $\sigma_{Q_2R}(\{g_i\}_{i=1}^4)$, and by performing entanglement swap operation $\mathcal S_{PQ_1Q_2R\to PR}$, a virtual link $\rho_{PR}(\{h_i\}_{i=1}^4)$ is created which stays in the four-parameter family with $h_i$s being calculated as 
    \begin{eqnarray}
        \nonumber h_k &=& 4(s_1+s_2+(-1)^{k-1}2s_3)\quad \text{for }k=1,2;\\
        \nonumber
        h_l &=& 4(s_4 + (-1)^l s_5)\quad \text{for }l=3,4.
    \end{eqnarray}
    Here, \( s_1 = \frac{1}{16},
   s_{2} = \frac{1}{16}[1 +2(f_1+f_2)][1 +2(g_1+g_2)],
   s_3 = \frac{1}{16}[(f_{1} - f_{2})(g_{1} - g_{2})],
   s_{4} = \frac{1}{8}(f_{3} + f_{4}),\) \text{and} \( 
 s_5 = \frac{1}{8}(f_{3} - f_{4})[1 +2(g_1+g_2)]\).
 \label{th:gen_swap_rule}
\end{theorem}
See Appendix~\ref{app:gen_reule} for the proofs of Theorem~\ref{th:arb_noise} and \ref{th:gen_swap_rule}.  

In summary, our analysis demonstrates that the most general state in AQN, consisting of an arbitrary number of nodes, is a four-parameter state in contrast with a single-parameter state in the case of the Pauli noise-affected scenario. Hence, we arrive at one of the main results of our work when the links are disturbed by the amplitude damping noise.
\begin{theorem}
    An arbitrary link of AQN can be represented as $h_1\Phi_{00}+h_2\Phi_{01}+h_3(\Phi_{01}^{00}+\Phi_{00}^{01})+2^{-1}(1-h_1-h_2)(\Phi_{10}+\Phi_{11})+h_4(\Phi_{10}^{11}+\Phi_{11}^{10})$ where \(h_i\)s can be written using the noise strength influencing each link.
    \label{th:gen_state_AQN}
\end{theorem}
This theorem proves that the number of parameters (four) in the states of the form $\sigma(\{h_i\}_{i=1}^4)$ is preserved under an arbitrary number of noise and swap rounds in AQN. Hence, Theorems~\ref{th:arb_noise} and \ref{th:gen_swap_rule} are indeed the general \textit{ noise update} and \textit{swap rule} in AQN.

\subsection{Fidelity and entanglement of links in AQN}
After obtaining the form of an arbitrary link in the presence of AD noise, it is natural to find its fidelity with the ideal situation in the absence of noise. In particular, an arbitrary link in AQN, represented by $\sigma(\{h_i\}_{i=1}^4)$ has the fidelity $\mathcal F (\sigma(\{h_i\}_{i=1}^4))=\bra{\Phi^+}\sigma(\{h_i\}_{i=1}^4)\ket{\Phi^+} = h_1$, with respect to the maximally entangled Bell state $\ket{\Phi^+}$. Further, the degree of entanglement, measured by the concurrence \cite{Wootters1998} of the state, can be calculated (see Appendix~\ref{app:concurence}) as
\begin{eqnarray}
    &&\nonumber \mathcal C(\sigma(\{h_i\}_{i=1}^4)) \\\nonumber&&= \max\bigg\{0,\bigg|\mathcal F(\sigma(\{h_i\}_{i=1}^4))-h_2\bigg|-\sqrt{(1-\sum_{i=1}^2h_i)^2-4h_4^2}\bigg\}.\\
    \label{eq:conc}
\end{eqnarray}
It is important to note from the above expressions that the entanglement of AQN-links is not uniquely determined by their fidelity which is in sharp contrast with the Pauli-noise setting (see Sec.~\ref{sec:pauli-twirled}). Specifically, when Pauli noise acts on the links, the amount of entanglement of the resulting Bell diagonal states can be uniquely determined by a single scalar parameter, i.e., the fidelity of the state with respect to $\ket{\Phi^+}$ state. Further, Bell diagonal states become entangled only when their fidelity with the target Bell state exceeds $1/2$ \cite{rohde2025quantuminternettechnicalversion}. This motivates us to determine the critical fidelity threshold for AQN that marks the onset of entanglement. It further indicates that amplitude damping noise-affected quantum network has to be benchmarked differently than the method used for the Pauli channels.

With all these regulations proven for AD noise, we are now ready to perform the simulation of an AQN using the noise update and the swap rule, derived in Theorems~\ref{th:arb_noise} and \ref{th:gen_swap_rule}, respectively. Before that, we take a brief digression to introduce the twirled ADC in order to set the stage for comparing the performance of twirled and non-twirled versions of ADC-affected networks.

{\textbf{Note}: By combining Theorem~\ref{th:gen_swap_rule} with Eq.~\eqref{eq:conc}, we observe that the entanglement of the swapped state depends on $h_4$, which is not invariant under the exchange $f_i \leftrightarrow g_i$. Consequently, the direction of classical communication plays a nontrivial role in the swapping operation of AQN. While we consistently adopt the forward classical communication (CC) strategy for applying Pauli corrections based on the outcomes of the BSM, in principle, one should optimize over the CC direction when performing a swap. In contrast, the fidelity remains symmetric under $f_i \leftrightarrow g_i$, thereby reinforcing the necessity of considering entanglement as another benchmark which was not essential for Pauli noise.}

\subsection{Comparing AQN and its Pauli-twirled version} 
\label{sec:pauli-twirled}
The action of the Pauli-twirled ADC (TADC) on a qubit $\varrho$ can be written as
\begin{eqnarray}
    \nonumber \overline{\mathcal N}(\varrho)&=&\frac{1}{4}\big[{\mathcal N}(\varrho) + X\mathcal N(X\varrho X)X + Y\mathcal N(Y\varrho Y)Y \\ &&+ Z\mathcal N(Z\varrho Z)Z\big].
\end{eqnarray}
Taking the coherence time of ADC as $m^\star$ (discrete time model), we can write $\gamma = 1-e^{-1/m^\star}$. Now, $\overline{\mathcal N}$ is a Pauli channel  \cite{Sarvepalli2009}, where
\begin{eqnarray}
    \overline{\mathcal N}(\varrho) = p_I\varrho+p_X X\varrho X + p_Y Y\varrho Y+p_Z Z\varrho Z,
\end{eqnarray}
with 
\begin{eqnarray}
    \nonumber p_I&=&\frac{1+e^{-\frac{1}{2m^\star}}}{2}-\frac{1-e^{-\frac{1}{m^\star}}}{4},\\\nonumber
    p_X &=& \frac{1-e^{-\frac{1}{m^\star}}}{4} = p_Y,\\
    p_Z &=& \frac{1-e^{-\frac{1}{2m^\star}}}{2}-\frac{1-e^{-\frac{1}{m^\star}}}{4}.
\end{eqnarray}

Here on, in a linear-chain quantum network, where we assume that instead of ADC, all the quantum memories are affected by the TADC ($\overline{\mathcal N}$), termed as TAQN. In this setting, the effect of TADC $n_1$ times on $\sigma=\Phi^+$ produces a state $\sigma(n_1)$ given by
\begin{eqnarray}
    \nonumber \sigma(n_1)&=&\frac{1}{4}\bigg[\bigl(1+e^{-2n_1/m^\star}+2e^{-n_1/m^\star}\bigr)\Phi_{00}\\\nonumber
&&+\bigl(1+e^{-2n_1/m^\star}-2e^{-n_1/m^\star}\bigr)\Phi_{01}
\\ &&+\bigl(1-e^{-2n_1/m^\star}\bigr)(\Phi_{10}
+\Phi_{11})\bigg].
\end{eqnarray}
Furthermore, the swap of two links $\sigma_{PQ_1}(n_1)$ and $\sigma_{Q_2R}(n_2)$, where the BSM is performed at memories $\{Q_1, Q_2\}$, generating a new link $\sigma_{PR}(n_1+n_2)$ \cite{schmidt2020, Khatri2022}. Clearly, noise acts additively during the swap operation in TAQN. 

\textit{Entanglement and fidelity.} The fidelity of the link $\sigma(n)$ is given by $\mathcal F(\sigma(n)) = \bigl(1+e^{-2n/m^\star}+2e^{-n/m^\star}\bigr)$. Moreover, the concurrence of the links in TAQN can be calculated as $\mathcal{C}(\sigma(n))=\max\{0,2\mathcal F(\sigma(n))-1\}$. Therefore, the entanglement of the PAQN-links is a linear function of fidelity and is non-vanishing only when $\mathcal{F} > 1/2$, as mentioned above.

\section{Comparison between entanglement distribution policies}
\label{sec:policies}

{In order to distribute entanglement dynamically and to perform numerical simulations to check efficiencies, one needs to include several decisions to be taken at each step of the protocol. This begins with the generation of elementary links, how long should links be kept active (cut-offs) to avoid too much decoherence, when and in which order entanglement swapping should be performed to connect shorter links to form longer virtual links, etc. A series of such actions taken on the quantum network constitutes an entanglement distribution policy. The question of finding optimal entanglement distribution policies has been addressed extensively in the literature \cite{Khatri2022,kamin2023, Iesta2023, guo2025} since it considerably affects the network performance both in terms of rate and fidelity of distributed states.
Nonetheless, as mentioned above, since most entanglement distribution literature have solely focused on Pauli noise, the policies have also been tailored to them. In contrast, we perform simulations for TAQN and AQN networks under some standard entanglement distribution policies and illustrate that their performance differs substantially. %\ph{This difference in performance suggests an exploration of unique policies for non-Pauli noise-affected networks, such as ADC.}

\begin{figure*}
    \centering
\includegraphics[width=0.8\linewidth]{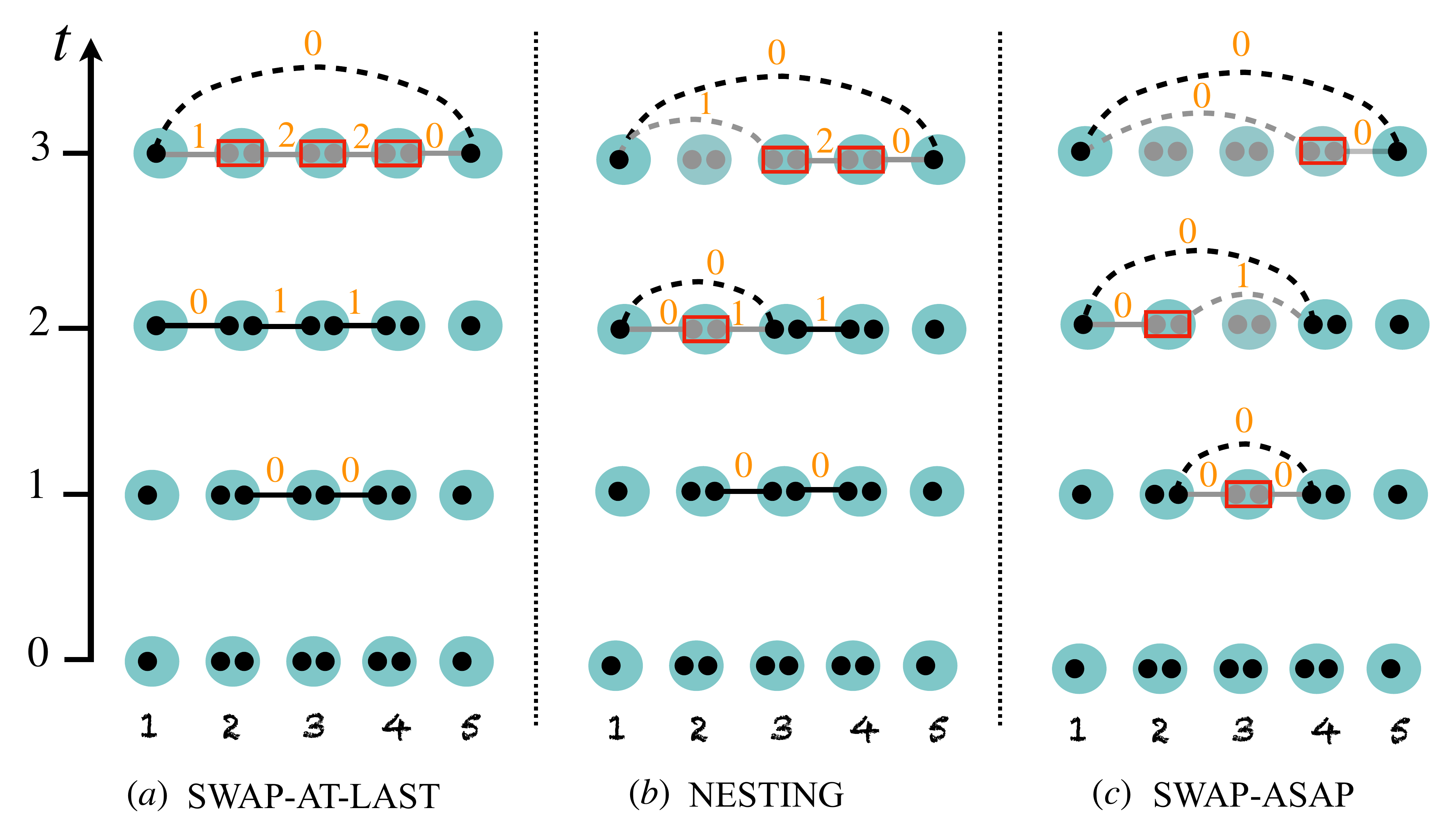}
    \caption{  
\textbf{An instance of three entanglement swapping policies,  \textsc{swap-at-last},  \textsc{nesting}, and  \textsc{swap-asap}, in a five-node linear AQN.} Each intermediate node holds two qubits, and all the swaps are deterministic and instantaneous. Link ages increase by one per time step, while virtual links created via a swap are reset to age \(``0"\). In contrast to PAQN, here, for each link, we need to keep track of both the four parameters of the state, along with its age.  A particular case of all successful link generation is shown, where the end-to-end link is generated through different orders of the swap operation. (a) \textsc{swap-at-last.}  The ages of all the elementary links are recorded till the end, i.e., up to the third time step. In the end, using the noise update rule (Theorem~\ref{th:arb_noise}) and swap rule (Theorem~\ref{th:gen_swap_rule}), the end-to-end link is created, having the age of this link reset to $0$. (b) \textsc{nesting.} Swaps are performed at those nodes where both the memories of the node have links of equal length. (c) \textsc{swap-asap.} The most efficient strategy is to execute a swap at any node as soon as both of its quantum memories are active. Notice that although the waiting time of all the policies is identical, the effect of decoherence on the quality of the end-to-end link should decrease as the policy is improved.
% At \(t=1\), links \(2\)-\(3\) and \(3\)-\(4\) (age \(0\)) form. At \(t=2\), link \(1\)-\(2\) (age \(0\)) is created, earlier links age to \(1\). At \(t=3\), link \(4\)-\(5\) (age \(0\)) is formed, and sequential swaps from nodes \(2\) to \(4\) produce the virtual link \(1\)-\(5\) (age \(0\)). (b) \textsc{nesting.} Here, swaps follow a on those node where left and right connection have same length. At \(t=2\), link \(1\)-\(2\) (age \(0\)) swaps with \(2\)-\(3\) (age \(1\)), yielding \(1\)-\(3\) (age \(0\)) while \(3\)-\(4\) ages to \(1\). At \(t=3\), link \(4\)-\(5\) (age \(0\)) is created and \(3\)-\(4\) (age \(2\)), \(1\)-\(3\) updates to age \(1\), and swaps at nodes \(4\) and \(3\) complete the end-to-end link \(1\)-\(5\) (age \(0\)). (c) \textsc{swap-asap.} The most efficient strategy is to execute swap at  Swaps occur immediately when adjacent links are present. At \(t=1\), links \(2\)-\(3\) and \(3\)-\(4\) (age \(0\)) form, and a swap at node \(3\) produces \(2\)-\(4\) (age \(0\)). At \(t=2\), link \(1\)-\(2\) (age \(0\)) is added, \(2\)-\(4\) updates to age \(1\), and a swap at node \(2\) yields \(1\)-\(4\) (age \(0\)). At \(t=3\), link \(4\)-\(5\) (age \(0\)) appears, \(1\)-\(4\) updates to age \(1\), and a swap at node \(4\) gives the final link \(1\)-\(5\) (age \(0\)).  
}
\label{fig:polices}
\end{figure*}

In this work, we consider three representative strategies: \textsc{swap-at-last}, \textsc{nesting}, and \textsc{swap-asap} (as depicted in Fig.~\ref{fig:polices}). In our setting, link generation is modeled as probabilistic, and the network is homogeneous, meaning all the elementary links are generated with the same probability \(p_l \in [0,1]\), while entanglement swapping is assumed to be deterministic. For simplicity, we do not include distillation or cutoff strategies in the present analysis. We also ignore classical communication costs, except for the heralding time for elementary link generation, which becomes the simulation time step in our model and thus sets the time scale. We primarily show the effect of different swapping strategies. 
This is motivated by the fact that, unlike single-parameter models such as the PQN, our AQN depends on four parameters, which makes the design of an optimal cutoff policy substantially more involved.
%{Thus, a systematic study and optimization of strategies is beyond the scope of this work \cite{mondal2026}}.

{\it \textsc{swap-at-last}.} In this strategy (Fig.~\ref{fig:polices}(a)), all elementary links between neighboring nodes are first generated probabilistically. Only after the full set of elementary links has been established, entanglement swapping is carried out simultaneously and deterministically at the intermediate nodes, yielding an end-to-end connection.  
Since swap operations are assumed to be instantaneous, no aging occurs during the swapping stage itself. However, every elementary link must wait until the last one is generated, and hence the earliest links in the chain accumulate significant {decoherence} by the time they are consumed. This effect can severely degrade the fidelity and entanglement of the end-to-end connection.

{\it \textsc{nesting}.} An alternative strategy is the \textsc{nesting} protocol, {originally introduced by Briegel et al. \cite{Briegel1998, Dur_1999}} which follows a hierarchical length-doubling procedure. Consider a repeater chain consisting of \(N = 2^n+1\) nodes connected by \(2^n\) elementary links with $n\in\mathbb Z_{>0}$, each of length \(L\). The protocol proceeds as follows:

\(1)\) Elementary links of length \(L\) are first generated probabilistically between neighboring nodes.

\(2)\) If two adjacent links of equal and required length (depending on the \textsc{nesting} level $n$ the node falls into) are simultaneously available at a node, an entanglement swap operation is performed there. Each such swap doubles the effective link length, producing successively longer connections: \(L \to 2L \to 4L \to \cdots \to 2^n L\).  
For a chain of \(2^n+1\) nodes, this process continues until level \(n\), where the final end-to-end link of length \(2^n L\) is obtained.

As a concrete example, consider the case shown in Fig.~\ref{fig:polices}(b), where $n=2$, and the chain consists of \(5\) nodes and \(4\) elementary links of length \(L\).  
At \textsc{nesting} level~1, swaps can occur at nodes \(2\) and \(4\) whenever they have links of length \(L\) on both sides; a swap is performed, producing links of length \(2L\).  
At \textsc{nesting} level~2, a swap at node \(3\) combines these links into the final end-to-end connection of length \(4L\).  
This illustrates the hierarchical doubling rule: at each level, swaps are performed at predetermined nodes that connect pairs of equally long links, until the entire chain is connected.  
Compared to \textsc{swap-at-last}, \textsc{nesting} reduces aging, since links are consumed as soon as their neighboring counterparts of equal and appropriate length become available. 

{\it \textsc{swap-asap}.} A third approach is the swap-as-soon-as-possible (\textsc{swap-asap}) protocol \cite{TJCoopmans2021,Shchukin2019, shchukin2022, kamin2023}. The procedure consists of the following steps:

\((1)\) Elementary links are generated probabilistically between neighboring nodes.

\((2)\) As soon as two adjacent links are simultaneously available, an entanglement swap operation is performed at their common node, regardless of whether the links are of equal length. There are no levels defined for swaps, and any node can perform a swap based on locally available information. Such a policy further reduces decoherence in the network by providing greater freedom in the sequence of swaps, especially in resource constrained scenarios such as low $m^\star$ and/or $p_l$. 

Here, we note that we restrict our simulations to deterministic entanglement swapping. Therefore, the average waiting time required to generate an end-to-end link remains the same for all the policies discussed above; it is the end-to-end fidelity and entanglement that vary depending on the efficiency of the policy, making them the relevant figures of merit for our simulations.

\subsection{Simulation methods: Tracking four-parameter states of AQN}
\label{sec:sim_method}
We employ Monte Carlo simulations to assess the performance of the linear-chain network influenced by amplitude damping noise. In all strategies, swap operations are assumed to be instantaneous and, therefore, do not contribute to the aging of a link. Instead, aging arises only during the waiting period between the generation of a link and its eventual consumption in a swap. Formally, if a link (elementary or virtual) is created at time step \(t_0\) and consumed at time step \(t\), its age is defined as \(t - t_0\) (see Fig.~\ref{fig:polices}). For the TAQN, we use the standard age-based simulation as described for example in Haldar et al. \cite{halder_stav}. Here, we need to track only the age of each link since swapping only leads to the addition of ages. For the AQN, we use the noise update rules and swap rules derived in Sec.~\ref {sec:main_proofs}. Here, the simulation adapts a more complex route. Each active link is represented by not only its age, but also the four parameters that represent its quantum state. The former is necessary to determine the noise updated parameters, whereas the latter allows for the calculation of the post entanglement swapping state. The age of a virtual link freshly produced from a swap is set to `$0$' (unlike the sum of ages of consumed links in the TAQN case), since the effect of noise inherited from the consumed links is already included in the four parameters of the virtual link via the swap rule. The next noise update occurs when this virtual link needs to be subsequently swapped, and the time between successive swaps becomes the age relevant to the noise update rule.
To obtain our figure of merits, such as end-to-end $(1)$ fidelity and $(2)$ entanglement, we average over $5\times 10^4$ independent evolutions of the network starting from the fully disconnected chain and ending at an end-to-end entangled state.

% \subsection{Simulation Results}
% \label{sec:sim_res}
% In the previous discussion (Sec.~\ref{sec:policies}), we described three distinct policies: \textsc{swap-at-last}, \textsc{nesting}, and \textsc{swap-asap}. We now employ these strategies to compare the performance of the AQN network with that of its ADC-twirled counterpart. This comparison highlights the role of noise modeling in long-distance entanglement distribution. For the case of a \(5\)-node network, our results indicate that the AQN network consistently outperforms the twirled-ADC network. Interestingly, within a certain range of link success probabilities \(p_l\) and coherence time $m^\star$, the AQN network exhibits an unbounded advantage in end-to-end entanglement generation when compared to the twirled-ADC model (see Fig.~\ref{fig:5node_mstar5}), since TAQN has zero entanglement compared to a finite value for AQN. \ph{Thus, our approach extends the parameter regime for which an entangled state can be generated in the end-to-end link.}

% In Fig.~\ref{fig:5node_mstar5}, we show the relative performance of the two $5$-node networks (TAQN and AQN) for different policies as a function of $p_l$. We fix the coherence time to \(m^\star=5\). 
\begin{figure}[h]
    \centering    \includegraphics[width=\linewidth]{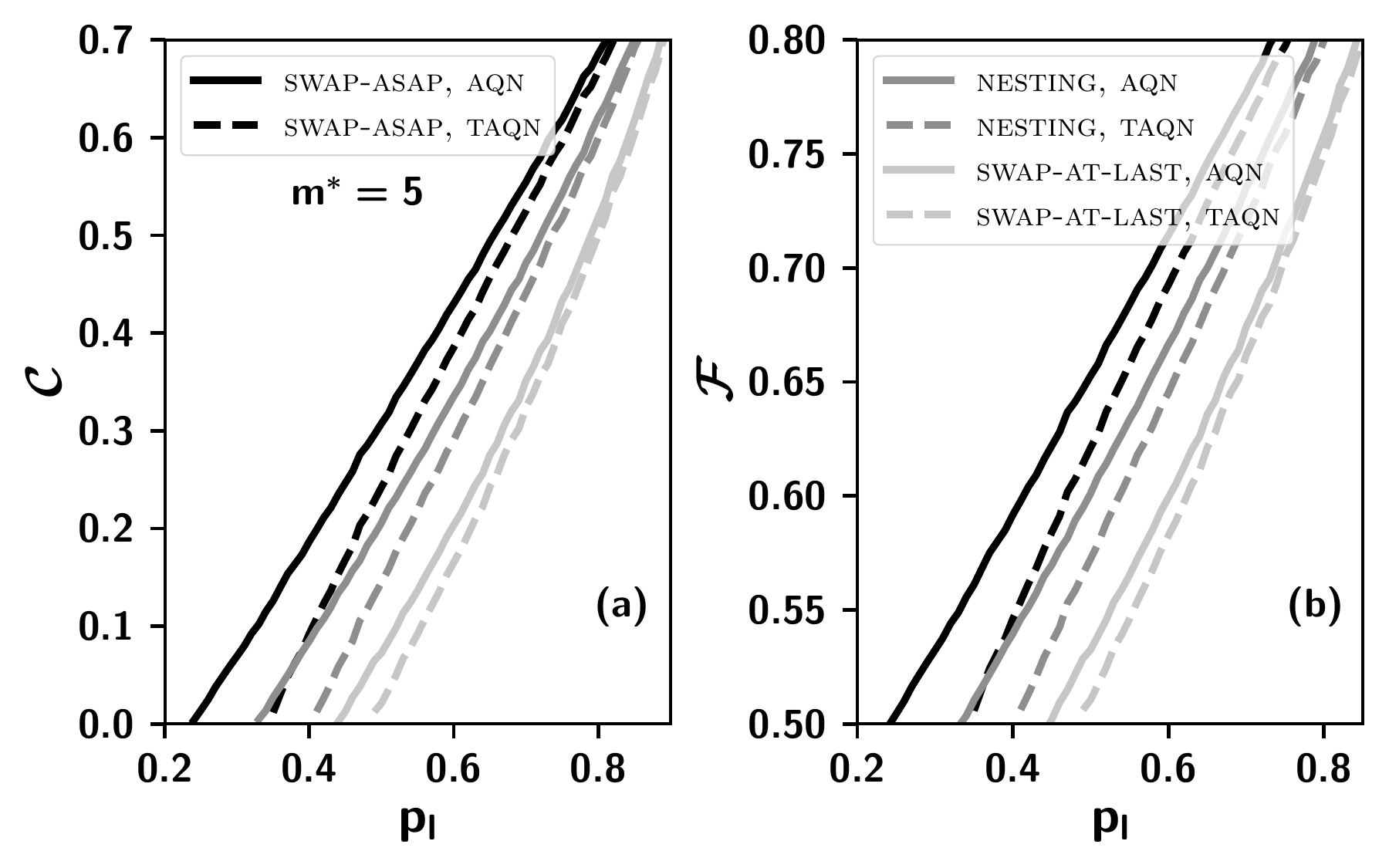}
    \caption{\textbf{Performance of entanglement and fidelity in a five-node repeater chain under \textsc{swap-asap}, \textsc{nesting}, and \textsc{swap-at-last} with fixed coherence time, ${m^\star=5}$.} 
Solid lines correspond to AQN and dashed lines to TAQN, with black for \textsc{swap-asap}, mid-grey for \textsc{nesting}, and light-grey for \textsc{swap-at-last}. 
(a) Concurrence, $\mathcal{C}$ (ordinate) vs. the link generation probability $p_l$ (abscissa): AQN attains non-vanishing entanglement at lower $p_l$ than TAQN. 
(b) Fidelity, $\mathcal{F}$ (ordinate) against $p_l$ (abscissa): AQN crosses the threshold $\mathcal{F}>\tfrac{1}{2}$ earlier than TAQN. 
For both $\mathcal C$ and $\mathcal F$, the performance gap between AQN and TAQN is most significant for \textsc{swap-asap}, moderate for \textsc{nesting}, and least for \textsc{swap-at-last}, highlighting that the advantage of simulating AQN is strongest when the elementary links are distributed under low $p_l$. The standard deviation of the average concurrence and fidelity of the end-to-end link is $\mathcal O(10^{-3})$ throughout our work. Both axes are dimensionless.}
    \label{fig:5node_mstar5}
\end{figure}

% We also evaluate the relative improvement in performance when using the exact noise model (AQN) versus when using the approximate twirled model (TAQN) for the same policy, which we fix to \textsc{swap-asap}. The relative improvement in end-to-end entanglement $\mathcal{I}$ is defined as:
% $$\mathcal{I} = \frac{\mathcal C_\text{AQN}-\mathcal C_\text{TAQN}}{\mathcal C_\text{TAQN}}\times 100 \%.$$ 
% This improvement as a function of $p_l$ and $m^\star$ is shown in Fig.~\ref{fig:5node_mstar_all}. It is noticeable that for entire parameter space explored $\mathcal{I}$ remains positive, implying an advantage for the AQN. The red region indicates infinite improvement, i.e. the TAQN has zero end-to-end entanglement whereas AQN can provide end-to-end entangled states. The white region indicates the region where neither AQN nor TAQN can deliver end-to-end entanglement. $\mathcal{I}$ is undefined in this region. It is also not a practically useful region of the parameter space.

\begin{figure}[h]
    \centering
    \includegraphics[width=\linewidth]{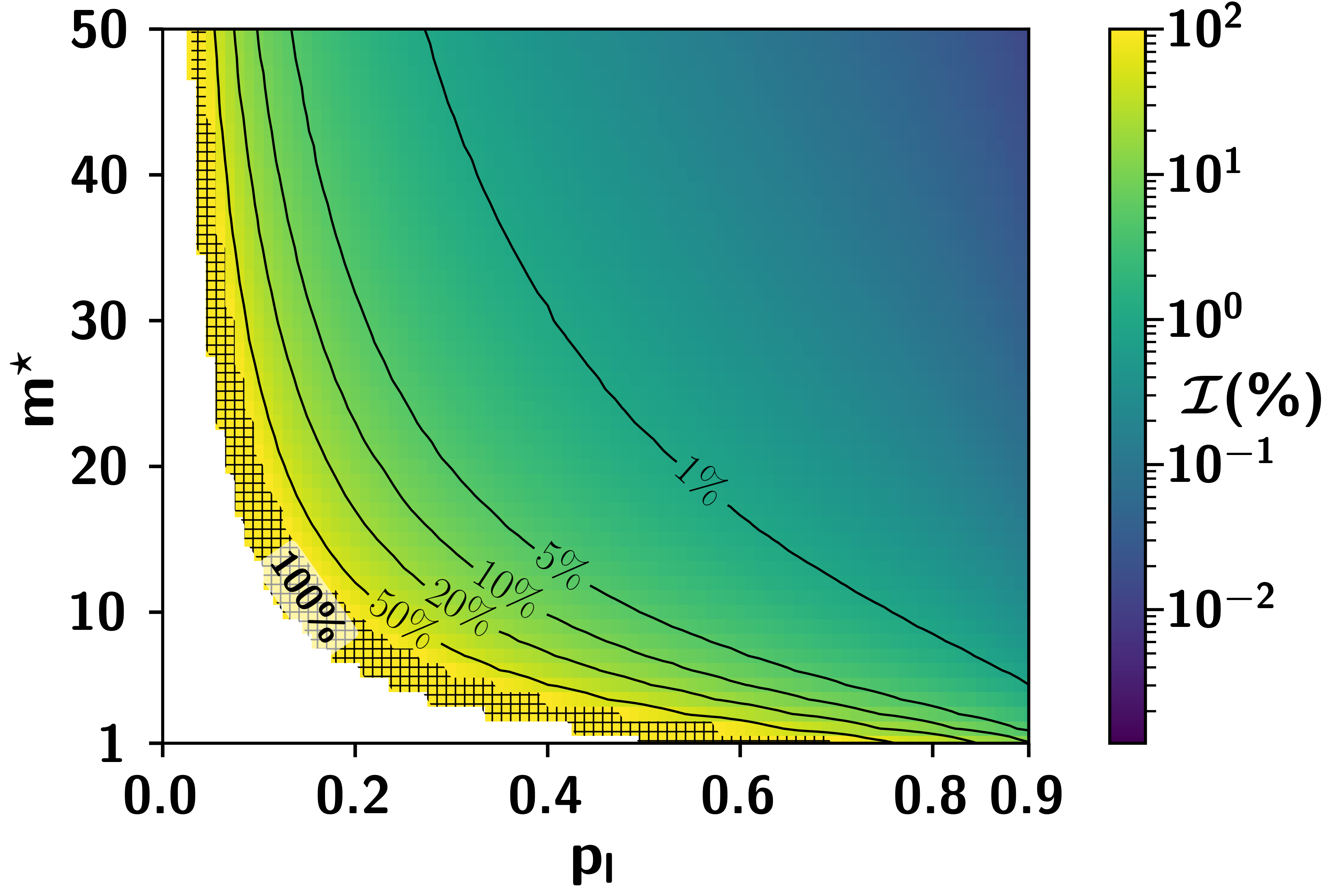}
    \caption{ \textbf{Comparison under the \textsc{swap-asap} policy between AQN and TAQN for a five-node repeater chain.}
The heatmap demonstrates the relative improvement factor,  \(\mathcal{I}\) as a function of coherence time \(m^*\) (vertical axis) and link generation probability \(p_l\) (horizontal axis). The shaded regions denote parameter regimes in which TAQN yields vanishing concurrence, whereas AQN attains a finite value. Consequently, this leads to an \textit{absolute-advantage}, since $\mathcal I = 100\%$ in these cases. As \(m^*\) and \(p_l\) increase, $\mathcal I$ gradually diminishes. Nevertheless, across the entire parameter space of \(m^*\) and \(p_l\), the improvement remains strictly positive. Both axes are dimensionless.}
    \label{fig:5node_mstar_all}
\end{figure}

\subsection{Simulation Results}
\label{sim_res}
We now employ three distinct strategies, \textsc{swap-at-last}, \textsc{nesting}, and \textsc{swap-asap} to compare the performance of the AQN network with that of its ADC-twirled counterparts. This comparison highlights the role of noise modeling in long-distance entanglement distribution. 

\textit{Five-node repeater chain.} We begin by fixing the number of nodes at $N=5$ and the coherence time at $m^\star=5$ (see Fig.~\ref{fig:5node_mstar5}). The performance of two five-node networks, AQN and TAQN, under different policies, is evaluated in terms of both entanglement and fidelity as a function of the link generation probability, $p_l$. Our results demonstrate that the AQN network consistently outperforms the TAQN network. As $p_l$ increases, the differences in concurrence and fidelity between the two networks decrease monotonically across all policies. This indicates that the advantage of performing an exact analysis of ADC, rather than employing its twirled approximation, is more pronounced in the regime of low $p_l$. Remarkably, for all policies considered, there exists a certain range of elementary link-generation probabilities in which the AQN network is able to establish an end-to-end entangled link, whereas the TAQN network fails to do so. 

 \begin{figure}[h]
    \centering
    \includegraphics[width=\linewidth]{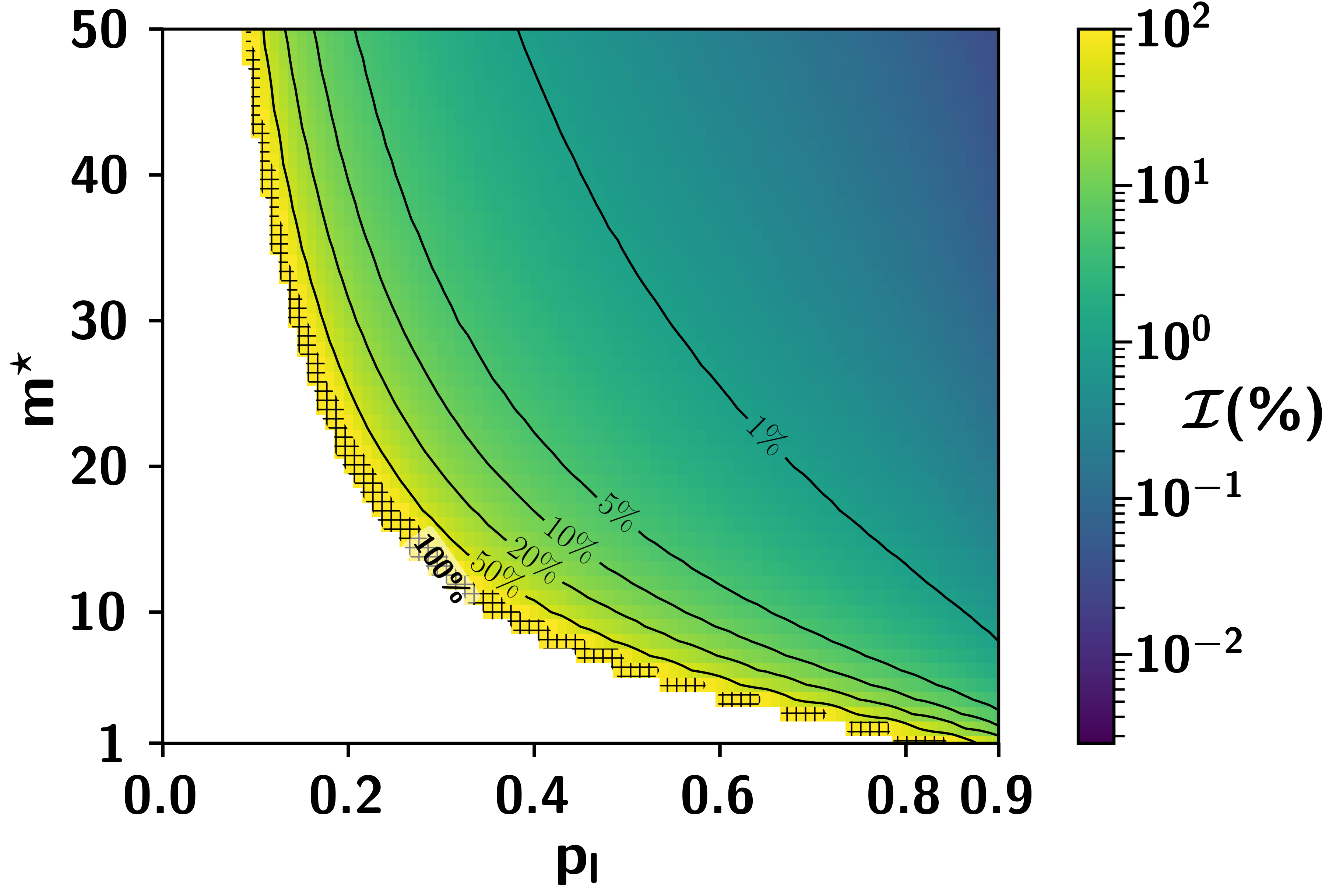}
    \caption{ \textbf{Performance of \textsc{swap-asap} policy under AQN and TAQN in a nine-node repeater chain.}  
All the configurations are the same as Fig.~\ref{fig:5node_mstar_all} except here $N=9$. The region of positive $\mathcal I$, along with the area of $\mathcal I=100\%$, is reduced compared to the five-node repeater. However, in the relevant region, for each instance of $p_l$ and $m^\star$, the value of $\mathcal I$ increases. This shows the scalability of our analysis. Both axes are dimensionless.}
    \label{fig:9node_inf}
\end{figure}

To quantitatively assess the improvements emerging from our approach, which is most discernible in our best policy, \textsc{swap-asap}, we define the relative improvement factor in the end-to-end entanglement as $$\mathcal{I} = \frac{\mathcal C_\text{AQN}-\mathcal C_\text{TAQN}}{\mathcal C_\text{AQN}}\times 100 \%.$$ This evaluates the relative improvement in the performance when using the exact noise model (AQN) versus when using the approximate twirled model (TAQN) for the same policy, which we fix to \textsc{swap-asap}. For a five-node linear chain we calculate $\mathcal I$ as a function of $p_l$ and $m^\star$ in Fig.~\ref{fig:5node_mstar_all}. It is noticeable that for the entire parameter space, $\mathcal{I}$ remains positive, implying a consistent advantage for the AQN. {The shaded region corresponds to the case of \(100\%\) relative improvement, \textit{representing an \textbf{absolute-advantage}, as the TAQN yields vanishing end-to-end entanglement whereas the AQN is capable of generating end-to-end entangled states.}} For shorter coherence times, this absolute advantage persists over a significantly wider range of $p_l$, which shrinks monotonically with increasing $m^\star$. For a given $m^\star$, as the success probability of link generation increases, the performance of TAQN and AQN converges. This analysis highlights that the exact treatment of AQN is particularly advantageous in the regime of low $p_l$ and low $m^\star$. The white region is not a practically useful region of parameter space, as it corresponds to scenarios where no end-to-end link is generated in both AQN and TAQN.

\textit{Requirement for better policy for large-scale networks to demonstrate improvement.} It is now natural to determine whether the advantage of AQN over TAQN reported before persists when the number of nodes is increased. To determine it, we consider a nine-node chain (see Fig.~\ref{fig:9node_inf}) and illustrate that in the regime $0\%\leq\mathcal I<100\%$, it increases for any given value of $m^\star$ and $p_l$ compared to a five-node chain. For instance, at $p_l=0.5, m^\star=10$, the corresponding values of $\mathcal I$ are $4.86\%$ and $17.714\%$, for $N=5$ and $N=9$, respectively. However, the parameter region exhibiting positive improvement shrinks substantially with the increase of nodes, including the shaded regime where an absolute-advantage is observed. \textcolor{black}{The latter can simply be explained by considering the increased noise accumulation in larger networks. As the size of the network increases, larger link generation probabilities and coherence times are needed to support end-to-end entanglement distribution for both AQN and TAQN. Nonetheless, as mentioned above, AQN clearly performs better in regions where end-to-end entanglement distribution is viable than the TAQN.} Further, it is important to view these results also in the context of results in Fig.~\ref{fig:5node_mstar5}. It was evident that the difference in performance increases as one chooses better entanglement distribution policies, namely being the least for \textsc{swap-at-last} and most for \textsc{swap-asap}. Previous works \cite{ Li_2021, kamin2023, halder_stav} have shown that the choice of optimal policies and cut-off strategies leads to even more improvement for longer chains. Considering the above two points, it could be inferred that the fall in fidelity, concurrence, and their improvement for AQN compared to TAQN is associated with the fact that we have not yet introduced any distillation procedure, cut-offs for fidelity control, nor is the order of the swapping optimal. Furthermore, there is a possibility of finding policies and protocols tailored to ADC noise.

%  \begin{figure}[h]
%     \centering
%     \includegraphics[width=\linewidth]{heatmap_9node_logscale_100.png}
%     \caption{ For node, \(N=9\) and decoherence time, \(m^\star\) vs \(p_l\). Both axes are dimensionless.}
%     \label{fig:9node_100}
% \end{figure}

%Swaps are deterministic, no cut off...

%\ph{For the improvement factor, we have calculated two quantities. Please tell which one is more logical and suitable for presentation? 
%\begin{itemize}
    %\item $\frac{\mathcal C_{ADC}-\mathcal C_{twADC}}{\mathcal C_{twADC}}$ where improvement factors are becoming very high and infinities, we have denoted by red points separately.

    %\item $\frac{\mathcal C_{ADC}-\mathcal C_{twADC}}{\mathcal C_{ADC}}$ where the improvement is scaled up to $100\%$ and exactly $100\%$ means the infinite improvement, i.e., twirled adc state is separable, as you know.
%\end{itemize}}

\section{Conclusion}
\label{sec:con}

Repeater-based quantum networks offer a promising approach to distribute entanglement across large quantum networks, with numerous optimal policies and cut-off strategies developed under the assumption of Pauli-type noise in quantum memories. However, a significant source of noise arises from energy-relaxation processes, effectively captured by amplitude damping (AD) noise. Earlier studies have treated this noise in quantum networks through twirling approximations, which unavoidably compromise the quality of the resulting end-to-end entangled link.

%Over the years, theoretical frameworks and several optimal policies, cut-off strategies have been developed to distribute entanglement in repeater-based quantum networks, where the quantum memories are suffering from Pauli-type noise. However, one of the main sources of noise comes from the relaxation of a higher energy level to a lower one, modeled by an amplitude damping noise. Previous works have addressed this noise model in quantum networks using twirling approximations, which inevitably compromise the quality of the end-to-end link.

Deviating from these simplified assumptions, we systematically developed a comprehensive theoretical framework for a homogeneous linear quantum repeater network influenced by amplitude damping noise, which we term an amplitude damping-affected quantum network (AQN).
%In this work, we carried out a systematic development of the theoretical framework of a homogeneous linear quantum repeater network subject to amplitude damping noise, denoted as AQN.
A key structural difference between AQN and its Pauli-twirled variant, the twirled amplitude damping-affected quantum network (TAQN), lies in the effective parametrization of the states as they undergo successive rounds of noise and entanglement swapping. In the AQN case, the states form a structured family, fully described by four real parameters, and remain block-diagonal in the Bell basis. Remarkably, this structure persists despite successive rounds of noise and swapping, confining the dynamics to this four-parameter family. 
%Specifically, in the Bell basis, the density matrix acquires a block-diagonal form, with one block spanning the pair $\{\Phi_{00}, \Phi_{01}\}$ and the other spanning $\{\Phi_{10}, \Phi_{11}\}$. 
In contrast, in the TAQN model, the states remain Bell-diagonal throughout the evolution, effectively characterized by a single parameter in terms of their age. While TAQN yields a simpler analytical description, it significantly restricts the accessible correlation space.
%Although TAQN provides a simplified analytical structure, it does so at the expense of restricting the accessible correlation space.
We developed the noise update and the swap rules for AQN, which are fundamentally different from TAQN, where the noise is additive under swap. We then utilized these rules to develop the simulation method of AQN by tracking these parameters along with the number of noise applications per link.
% Here, we tracked the four parameters of a link along with the number of times noise acts on it. 

Our numerical and analytical results demonstrated that the structural richness of AQN states translates into a clear performance advantage over TAQN. Specifically, AQN consistently outperforms TAQN with respect to both the entanglement and the fidelity of the end-to-end shared state. The advantage is particularly striking in the regime of low probability of elementary link generation and low coherence time, which is highly relevant for near-term experiments. We found that AQN  can exhibit, in effect,  improvements over TAQN, underscoring the practical significance of retaining the full density matrix structure rather than twirling it away. The role of various strategies of entanglement distribution accentuates this distinction. Among the policies studied -- \textsc{swap-at-last}, \textsc{nesting}, and \textsc{swap-asap} -- the performance hierarchy is clear: \textsc{swap-asap} yields the best performance, followed by \textsc{nesting}, with \textsc{swap-at-last} performing the least effectively. Crucially, the relative advantage of AQN over TAQN becomes more pronounced as the restriction on the swapping strategy is reduced. However, the region of improvement gradually decreases with an increase in the number of nodes. This is because larger quantum networks require more efficient policies {along with distillation  \cite{siddhu2025basicdistillationrealisticnoise}}, and the introduction of cut-off strategies to reduce the effect of decoherence, which we aim to address in future work \cite{mondal2026}.

% This indicates that the additional correlations preserved in AQN states are most fully exploited under strategies that aggressively propagate entanglement through the network.  
In summary, our results strongly advocate for adopting AQN as a
%Taken together, these findings established a compelling case for considering AQN as the
more faithful and beneficial model than TAQN for addressing realistic amplitude damping noise in quantum repeater chains. 
%By retaining a richer state structure while remaining analytically tractable, AQN not only offers deeper physical insight but also points toward practical improvements in end-to-end entanglement distribution under experimentally relevant conditions.  
Unlike TAQN, AQN retains a richer state structure while remaining analytically manageable, thereby offering a deeper theoretical understanding. At the same time, it provides practical advantages by enabling improved end-to-end entanglement distribution across repeater networks, particularly under conditions relevant to near-term experimental implementations and real-world quantum communication technologies.

\acknowledgements
We acknowledge the use of \href{https://github.com/titaschanda/QIClib}{QIClib} -- a modern C++ library for general purpose quantum information processing and quantum computing (\url{https://titaschanda.github.io/QIClib}) and cluster computing facility at Harish-Chandra Research Institute. PH acknowledges ``INFOSYS
scholarship for senior students". We acknowledge support from the project entitled ``Technology Vertical - Quantum Communication'' under the National Quantum Mission of the Department of Science and Technology (DST)  (Sanction Order No. DST/QTC/NQM/QComm/$2024/2$ (G)).

\onecolumngrid
\appendix
\section{Expressing Bell States in the Pauli Basis and their transformation under ADC}
\label{app:noise_on_pauli_basis}

We analyze the effect of the local amplitude damping channel (ADC) \cite{Nielsen2000QuantumComputation}, defined as $\mathcal{N}(\star) = \sum_{i=0}^{1} K_{i} (\star) K_{i}^{\dagger},$ with $K_0=\ketbra0+\sqrt{1-\gamma}\ketbra1$ and $K_1=\sqrt\gamma\ketbra{0}{1}$ where $\gamma\in[0,1]$ is the noise strength. In a discrete-time setting, the noise parameter can be related to the coherence time $m^\star$ of the channel as $\gamma=1-\exp(-1/m^\star)$ where $m^\star\in\{0,1,2,\ldots,\infty\}$ \cite{Ghosh2012}. The decoherence effect by ADC after $m_1$ time steps on a single qubit, 
\begin{eqnarray}
  \rho=  \begin{bmatrix}
1-\alpha & \beta \\
\beta^\star & \alpha
\end{bmatrix},
\end{eqnarray}
can be expressed as
\begin{eqnarray}
  \mathcal N^{\circ m_1}(\rho)=  \begin{bmatrix}
1-\alpha e^{-\frac{m_1}{m^\star}} & \beta e^{-\frac{m_1}{2m^\star}} \\
\beta^\star e^{-\frac{m_1}{2m^\star}} & \alpha e^{-\frac{m_1}{m^\star}}
\end{bmatrix},
\end{eqnarray}
where $\mathcal N^{\circ m_1}$ denotes the concatenation of $m_1$ number of ADC.

Given an fresh elementary link $\sigma_{PQ}=\mathtt{Proj}(\ket{\Phi^+})$  if the qubit memories $P$ and $Q$ decohere for $m_1$ time steps, then the noise updated link is written as
\begin{eqnarray}
    \nonumber \sigma_{PQ}(m_1)&=&(\mathcal N_P^{\circ {m_1}}\otimes \mathcal N_{Q}^{{\circ m_1}})\sigma_{PQ}.
\end{eqnarray}
Now, the four Bell states are defined as $\ket{\Phi_{xz}}=(\mathbb 1\otimes X^xZ^z)\ket{\Phi^+}$ for all $x,z\in \{0,1\}$, with $\ket{\Phi^{\pm}}=(\ket{00}\pm\ket{11})/\sqrt{2}$, $\ket{\Psi^{\pm}}=(\ket{01}\pm\ket{10})/\sqrt{2}$. One can easily obtain that $\ket{\Phi_{00}}=\ket{\Phi^+},\ket{\Phi_{01}}=\ket{\Phi^-},\ket{\Phi_{10}}=\ket{\Psi^+},$ and $\ket{\Phi_{11}}=\ket{\Psi^-}$. Here, $\{\mathbb 1, X, Y ,Z\}$ is the Pauli basis. The Bell states can be alternatively written in terms of the Pauli basis as 
\begin{eqnarray}
    \nonumber \Phi_{00}&=&\frac{1}{4}(\mathbb 1\otimes\mathbb 1+X\otimes X-Y\otimes Y+Z\otimes Z),\\
    \nonumber \Phi_{01}&=&\frac{1}{4}(\mathbb 1\otimes\mathbb 1-X\otimes X+Y\otimes Y+Z\otimes Z),\\
    \nonumber \Phi_{10}&=&\frac{1}{4}(\mathbb 1\otimes\mathbb 1+X\otimes X+Y\otimes Y-Z\otimes Z),\\
    \nonumber \Phi_{11}&=&\frac{1}{4}(\mathbb 1\otimes\mathbb 1-X\otimes X-Y\otimes Y-Z\otimes Z),
\end{eqnarray}
where $\Phi_{xz}=\ket{\Phi_{xz}}\bra{\Phi_{xz}}$. Action of amplitude damping noise on the Pauli basis is written as
\begin{eqnarray}
    \mathcal{N}^{\circ m_{1}}(\mathbb{1}) &=& \mathbb{1} + (1-a)Z, \quad 
    \mathcal{N}^{\circ m_{1}}(X) = \sqrt{a}X, \nonumber \\
    \mathcal{N}^{\circ m_{1}}(Y) &=& \sqrt{a}Y, \quad 
    \mathcal{N}^{\circ m_{1}}(Z) = aZ,
\end{eqnarray}
where $a=(1-\gamma)^{m_1}$. By further calculation, we arrive at the expression
\begin{eqnarray}
    \nonumber \sigma_{PQ}(m_1)&=&\frac{1}{4}[\mathbb 1\otimes \mathbb 1+a( X\otimes  X- Y\otimes  Y)\\\nonumber &&+(1- a)(\mathbb 1\otimes  Z+ Z\otimes \mathbb 1)\\&&+\{(1- a)^2+ a^2\} Z\otimes  Z].
    \label{eq:m1_times_noise}
\end{eqnarray}
Therefore, by changing Eq.~(\ref{eq:m1_times_noise}) in the Bell basis, one can obtain
\begin{eqnarray}
\sigma_{PQ}&=&c_1\Phi_{00}+c_2\Phi_{01}+c_3(\Phi_{01}^{00}+\Phi_{00}^{01})+2^{-1}(1-\sum_{i=1}^2c_i)(\Phi_{10}+\Phi_{11}),
    \label{eq:m1_noise_bell}
\end{eqnarray}
where, \(\mathbb{1}\otimes Z+Z\otimes\mathbb1=2\big(\ket{\Phi_{00}}\bra{\Phi_{01}}+\ket{\Phi_{01}}\bra{\Phi_{00}}\big)\), and \(Z\otimes\mathbb1-\mathbb{1}\otimes Z=2\big(\ket{\Phi_{10}}\bra{\Phi_{11}}+\ket{\Phi_{11}}\bra{\Phi_{10}}\big)\),\\\(\quad c_1=\frac{(1+a^2)}{2},\quad  c_2=\frac{(1-a)^2}{2}, \quad c_3=\frac{(1-a)}{2}\).

\section{Swap update rule of two elementary noisy links}
\label{app:swap_1st_round}
\begin{comment}
The four Bell state is defined as $\ket{\Phi_{xz}}=(\mathbb 1\otimes X^xZ^z)\ket{\Phi^+}$ for all $x,z\in \{0,1\}$. Here, $\{X, Y ,Z\}$ is the set of Pauli matrices. The Bell states can be alternatively written in terms of Pauli matrices as 
\begin{eqnarray}
    \nonumber {\Phi_{00}}&=&\frac{1}{4}(\mathbb 1\otimes\mathbb 1+X\otimes X-Y\otimes Y+Z\otimes Z),\quad {\Phi_{01}}=\frac{1}{4}(\mathbb 1\otimes\mathbb 1-X\otimes X+Y\otimes Y+Z\otimes Z),\\
    {\Phi_{10}}&=&\frac{1}{4}(\mathbb 1\otimes\mathbb 1+X\otimes X+Y\otimes Y-Z\otimes Z),\quad {\Phi_{11}}=\frac{1}{4}(\mathbb 1\otimes\mathbb 1-X\otimes X-Y\otimes Y-Z\otimes Z).
\end{eqnarray}
\end{comment}
Let us consider two links \(\rho_{PQ_1}=\Phi_{00}\) and \(\rho_{Q_2R}=\Phi_{00}\) where  entanglement swapping takes place at the node \((Q_1Q_2)\). The action of ADC $m_1$ times on a given link $PQ_1$ can be represented as
\begin{eqnarray}
    \nonumber \sigma_{AR_1}(m_1)&=&c_1\Phi_{00}+c_2\Phi_{01}+c_3(\Phi_{01}^{00}+\Phi_{00}^{01})+c_4(\Phi_{10}+\Phi_{11}),
    \label{eq:noisy_1st_link}
\end{eqnarray}
with 
\begin{eqnarray}
   \nonumber  c_1&=&\frac{(1+a^2)}{2},\quad c_2=\frac{(1-a)^2}{2},\quad c_3=\frac{(1-a)}{2},\quad c_4=\frac12(1-\sum_{i=1}^2c_i)=\frac{a(1-a)}{2}.
   \label{eq:update_1st}
\end{eqnarray}
Similarly, after \(m_2\) times local ADC noise on \(Q_2R\) link can be  given as
\begin{eqnarray}
 \nonumber \sigma_{Q_2R}(m_2)&=&d_1\Phi_{00}+d_2\Phi_{01}+d_3(\Phi_{01}^{00}+\Phi_{00}^{01})+d_4(\Phi_{10}+\Phi_{11}),
 \label{eq:noisy_2nd_link}
\end{eqnarray}
with 
\begin{eqnarray}
   \nonumber  d_1&=&\frac{(1+b^2)}{2},\quad d_2=\frac{(1-b)^2}{2},\quad d_3=\frac{(1-b)}{2},\quad d_4=\frac12(1-\sum_{i=1}^2d_i)=\frac{b(1-b)}{2},
   \label{eq:update_2nd}
\end{eqnarray}
where \(a=(1-\gamma)^{m_1}\) and \(b=(1-\gamma)^{m_2}\).
Here, $\Phi_{xz}=\ketbra{\Phi_{xz}}{\Phi_{xz}}$ and $\Phi_{x'z'}^{xz}=\ketbra{\Phi_{xz}}{\Phi_{x'z'}}$. Therefore, the composite state of two links is given by 
\begin{eqnarray}   \rho_{PQ_1Q_2R}=\sigma_{PQ_1}(m_1)\otimes \sigma_{Q_2R}(m_2).
\end{eqnarray}
The relevant terms for swapping on node \(Q_1Q_2\) are given by
\begin{eqnarray}
   \sigma_{PQ_1}(m_1)\otimes \sigma_{Q_2R}(m_2)& \rightarrow& \nonumber\frac{1}{16}(\sum_{i=1}^4\mathcal{S}_i)_{PRQ_1Q_2},\\ 
   &=&\frac{1}{16}[D(0,0,1,0)\Phi_{00}
   +D(0,1,0,0)\Phi_{01}
   +D(0,0,0,0)\Phi_{10}
   +D(0,1,1,1)\Phi_{11}],
   \label{eq:swap_rel}
\end{eqnarray}
where \(\mathcal{S}_1=L_{a}\otimes L_{b}\otimes \mathbb{1}\otimes \mathbb{1}\), \(\mathcal{S}_2=ab  
 X^{\otimes 4} \), \(\mathcal{S}_3=ab  Y^{\otimes 4}\) and \(\mathcal{S}_4=M_{a}\otimes M_{b}\otimes Z^{\otimes 2} \) and  \(D(l,m,n,q)=(-1)^lL_{a}\otimes L_{b} +(-1)^m ab X^{\otimes 2}  +(-1)^n ab Y^{\otimes 2} +(-1)^q M_{a}\otimes M_{b}\), and \(L_{a}=\mathbb{1}+(1-a)Z\), \(L_{b}=\mathbb{1}+(1-b)Z\) and \(M_{a}=(1-a)\mathbb{1}+\mathbb{a}Z\), \(M_{b}=(1-b)\mathbb{1}+\mathbb{b}Z\). Here, \(\mathbb{a}=a^2+(1-a)^2\) and  \(\mathbb{b}=b^2+(1-b)^2\).

After the entanglement swapping protocol, we have the average state as
\begin{eqnarray}
    \nonumber &&\rho_{PR}=\mathcal S_{PRQ_1Q_2\to PR}(\rho_{PRQ_1Q_2})\\\nonumber &=&\sum_{x,z=0}^1(\bra{\Phi_{xz}}_{Q_1Q_2}\otimes X^x_RZ^z_R)\rho_{PRQ_1Q_2}(\ket{\Phi_{xz}}_{Q_1Q_2}\otimes Z^z_RX^x_R),\\\nonumber
    &=& [f_1\Phi_{00}+f_2\Phi_{01}+f_3(\Phi_{01}^{00}+\Phi_{00}^{01}) + f_4(\Phi_{11}^{10}+\Phi_{10}^{11})+f_5(\Phi_{10}+\Phi_{11})]_{PR},
    \label{eq:1st_swap}
\end{eqnarray}
where
\begin{eqnarray}
  \nonumber f_1&=&\frac{1}{4}(1+2ab+\mathbb{ab}),\quad f_2=\frac{1}{4}(1-2ab+\mathbb{ab}),\quad f_3=\frac{1}{4}(1-a)(1+\mathbb b),\quad f_4=\frac{1}{4}(1-a)(1-\mathbb b),\nonumber\\&& f_5=\frac12(1-\sum_{i=1}^2f_i)=\frac{1}{4}(1- \mathbb a \mathbb b).
   \label{eq:1st_swap_all_coefficient}
\end{eqnarray}
Therefore, by replacing \(a\) and \(b\), \(\mathbb{a}=a^2+(1-a)^2\) and \(\mathbb{b}=b^2+(1-b)^2\) in Eq.~(\ref{eq:1st_swap_all_coefficient}) in terms of \(\{c_i,d_i\}\), the coefficients takes the form as  
\begin{eqnarray}
  \nonumber f_{1}&=&\frac{1}{4}\bigg(1+ \big(2\sum_{i=1}^2c_i-1\big)\big(2\sum_{i=1}^2d_i-1\big) +2 (c_1-c_2)(d_1-d_2)\bigg),\\\nonumber  f_{2}&=&\frac{1}{4}\bigg(1+ \big(2\sum_{i=1}^2c_i-1\big)\big(2\sum_{i=1}^2d_i-1\big) -2 (c_1-c_2)(d_1-d_2)\bigg),\\\nonumber f_3&=&c_3\sum_{i=1}^2d_i,\quad f_4=c_3(1-\sum_{i=1}^2d_i),\quad f_5=\frac12(1-\sum_{i=1}^2f_i)=\frac{1}{4}\bigg(1-\big(2\sum_{i=1}^2c_i-1\big)\big(2\sum_{i=1}^2d_i-1\big)\bigg).
   \label{eq:1st_swap_all_coe_1}
\end{eqnarray}
\section{General age update rule in network under amplitude damping noise}
\label{app:gen_reule}

\textbf{Action of ADC on links:} Under ADC, the most generalized state in the entanglement swapping scenario is given by 
\begin{eqnarray}
 \rho_{PQ_1}&=&f_1\Phi_{00}+f_2\Phi_{01}+f_3(\Phi_{01}^{00}+\Phi_{00}^{01})+ f_4(\Phi_{11}^{10}+\Phi_{10}^{11})]_{AR_1}+f_5(\Phi_{10}+\Phi_{11}),  
\end{eqnarray}
with the normalization constant, $f_1+f_2+2f_5=1$. Similarly, the link between $Q_2$ and $R$ is given by 
\begin{eqnarray}
    \nonumber \rho_{Q_2R}&=&g_1\Phi_{00}+g_2\Phi_{01}+g_3(\Phi_{01}^{00}+\Phi_{00}^{01})+ g_4(\Phi_{11}^{10}+\Phi_{10}^{11})+g_5(\Phi_{10}+\Phi_{11}),
\end{eqnarray}
with $g_1+g_2+2g_5=1$.
The action of ADC on the link $PQ_1$ is given by
\begin{eqnarray}
    \nonumber (\mathcal N_P^{n_1}\otimes \mathcal N_{Q_1}^{n_1})\rho_{PQ_1}&=&[f_1^1\Phi_{00}+f_2^1\Phi_{01}+f_3^1(\Phi_{01}^{00}+\Phi_{00}^{01})
    + f_4^1(\Phi_{11}^{10}+\Phi_{10}^{11})+f_5^1(\Phi_{10}+\Phi_{11})]_{PQ_1},\\
\end{eqnarray}
where 
\begin{eqnarray}
    \nonumber f_1^1&=&\frac{1}{4}[f_1(1+2 t+\mathbb t_+)+f_2(1-2t+\mathbb t_+) +2f_5(1+\mathbb t_-)+4f_3t(1-t)],\\\nonumber f_2^1&=&\frac{1}{4}[f_1(1-2t+\mathbb t_+)+f_2(1+2t+\mathbb t_+) +2f_5(1+\mathbb t_-)+4f_3t(1-t)],\\\nonumber
    f_3^1&=&\frac{1}{2}(1-t)+tf_3,\\ \nonumber f_4^1&=&f_4t,\\
    f_5^1&=&\frac{1}{4}[(f_1+f_2)(1-\mathbb t_+)+2f_5\big(1-\mathbb{t_-}\big) -4f_3t(1-t)],
    \label{eq:adc_general_noise_update}
\end{eqnarray}
with \(t=(1-\gamma)^{n_1}\),  \(\mathbb{t_\pm}=t^2\pm(1-t)^2\) and the normalization constant, \(f_5^1=\frac12\big(1-\sum_{i=1}^2f_{i}^1\big)\). Similarly, the effect of ADC on the link $Q_2R$ can be expressed as 
\begin{eqnarray}
    \nonumber (\mathcal N_{Q_{2}}^{n_2}\otimes \mathcal N_{R}^{n_2})\rho_{Q_{2}R}&=&[g_1^1\Phi_{00}+g_2^1\Phi_{01}+g_3^1(\Phi_{01}^{00}+\Phi_{00}^{01})+g_4^1(\Phi_{11}^{10}+\Phi_{10}^{11})+g_5^1(\Phi_{10}+\Phi_{11})]_{Q_2R}.
\end{eqnarray}
Here, the updation of $g_i\to g_i^1$ $\forall i$ follows the similar rule of Eq.~\eqref{eq:adc_general_noise_update} with $t$ and $\mathbb t$ being replaced by $v=(1-\gamma)^{n_2}$ and $\mathbb v=v^2+(1-v)^2$ respectively and the normalization constant, \(g_5^1=\frac12\big(1-\sum_{i=1}^2g_{i}^1\big)\). Therefore, the composite state of two links is given by 
\begin{eqnarray}   \rho_{PQ_1Q_2R}=(\mathcal N_P^{n_1}\otimes \mathcal N_{Q_1}^{n_1})\rho_{PQ_1}\otimes (\mathcal N_{Q_2}^{n_2}\otimes \mathcal N_{R}^{n_2})\rho_{Q_2R}.
\end{eqnarray}

\textbf{Swapping:} After the entanglement swapping protocol, we have the average state as 
\begin{eqnarray}
    \nonumber &&\rho_{PR}=\mathcal S_{PQ_1Q_2R\to PR}(\rho_{PQ_1Q_2R})\\\nonumber &=&\sum_{x,z=0}^1(\bra{\Phi_{xz}}_{Q_1Q_2}\otimes X^x_RZ^z_R)\rho_{PQ_1Q_2R}(\ket{\Phi_{xz}}_{Q_1Q_2}\otimes Z^z_RX^x_R),\\\nonumber
    &=& [h_1\Phi_{00}+h_2\Phi_{01}+h_3(\Phi_{01}^{00}+\Phi_{00}^{01})+ h_4(\Phi_{11}^{10}+\Phi_{10}^{11})+h_5(\Phi_{10}+\Phi_{11})]_{PR},
\end{eqnarray}
and the normalization constant, \(h_1+h_2+2h_5=1\), where 
\begin{eqnarray}
   \nonumber h_1&=&4(s_{1}+s_{2}+2s_{3}),\quad h_2=4(s_{1}+s_{2}-2s_{3}),\\\nonumber h_3&=&4(s_{1}-s_{2}),\quad h_4=4(s_{4}+s_{5}),\\
   h_5&=&\frac12\big(1-\sum_{i=1}^2h_i\big)=4(s_{4}-s_{5}),
   \label{eq:gen_swap_update}
\end{eqnarray}
with 
\begin{eqnarray}
   \nonumber s_1 &=& \frac{1}{16}, \quad 
   s_{2} = \frac{1}{16}[(1 - 4f_{5}^1)(1 - 4g_{5}^1)] = \frac{1}{16}[1 +2(f_1^1+f_2^1)][1 +2(g_1^1+g_2^1)], \\\nonumber \quad 
   s_3 &=& \frac{1}{16}[(f_{1}^1 - f_{2}^1)(g_{1}^1 - g_{2}^1)],\quad 
   s_{4} = \frac{1}{8}(f_{3}^{1} + f_{4}^{1}),\\ 
   \nonumber s_5 &=& \frac{1}{8}[(f_{3}^1 - f_{4}^1)(1 - 4g_{5}^1)].
\end{eqnarray}

\section{Entanglement of links in AQN}
\label{app:concurence}
For two-qubit systems, entanglement can be quantified using concurrence~\cite{Wootters1998} and is defined as
\begin{eqnarray}
    C(\rho) = \max\left\{0, \sqrt{\lambda_1} - \sqrt{\lambda_2} - \sqrt{\lambda_3} - \sqrt{\lambda_4}\right\},
    \label{eq:concurrence}
\end{eqnarray}
where $\lambda_i$s are the eigenvalues, in decreasing order, of the matrix 
$\tilde{\rho} = \rho \, (\sigma_y \otimes \sigma_y) \, \rho^\star \, (\sigma_y \otimes \sigma_y)$.
In this work, we focus on the important class of two-qubit $X$-states, whose density matrices have 
nonzero entries only along the main diagonal and anti-diagonal,
\begin{equation}
\rho^{AB} =
\begin{pmatrix}
a & 0 & 0 & w \\
0 & b & z & 0 \\
0 & z^* & c & 0 \\
w^* & 0 & 0 & d
\end{pmatrix}, 
\qquad a+b+c+d = 1.
\end{equation}
The concurrence, in this case, becomes \cite{Yu2007}
\(C(\rho) = 2 \max \big( 0,\, |z| - \sqrt{ad},\, |w| - \sqrt{bc} \big).
\) An arbitrary link of AQN can be represented as
\(\rho_{f}=h_1\Phi_{00}+h_2\Phi_{01}+h_3(\Phi_{01}^{00}+\Phi_{00}^{01})+h_4(\Phi_{10}^{11}+\Phi_{11}^{10})+h_5(\Phi_{10}+\Phi_{11})\) (see Theorem~\ref{th:gen_state_AQN}) where \(h_5=\frac{1}{2}(1-h_1-h_2)\) and in the computational basis, it takes the form as
\begin{equation}
\rho_{f} =
\begin{pmatrix}
\frac{h_1+h_2}{2}+h_3 & 0 & 0 & \frac{h_1-h_2}{2} \\
0 & h_5+h_4 & 0 & 0 \\
0 & 0 & h_5-h_4 & 0 \\
\frac{h_1-h_2}{2} & 0 & 0 & \frac{h_1+h_2}{2}+h_3
\end{pmatrix}.
\label{eq:con_mat}
\end{equation}
Therefore, the concurrence of \(\rho_f\) is 
\begin{eqnarray}
    C(\rho_f)&=&2\max \big(0,|\frac{h_1-h_2}{2}|-\sqrt{h_5^2-h_4^2}\big),\nonumber\\&=&\max\bigg\{0,\bigg|\mathcal F(\sigma(\{h_i\}_{i=1}^4))-h_2\bigg|-\sqrt{(1-\sum_{i=1}^2h_i)^2-4h_4^2}\bigg\}.
\end{eqnarray}
where \(h_1=\mathcal F(\sigma(\{h_i\}_{i=1}^4))\) is the fidelity of end-to-end link.

% \section{Numerical analysis: 17-nodes networks via \textsc{swap-asap}}
% \label{app:17_nodes_swap_asap}
% \begin{figure}[h]
%     \centering
%     \includegraphics[width=0.6\linewidth]{heatmap_17node.png}
%     \caption{ For node, \(N=17\) and decoherence time, \(m^\star\) vs \(p_l\). Both axes are dimensionless.}
%     \label{fig:17node_inf}
% \end{figure}
\twocolumngrid

\bibliography{reference}

\end{document}